\documentclass[11pt]{article} % longversion and shortversion

\usepackage[margin=1.3in]{geometry}
\usepackage[utf8]{inputenc}
\usepackage[T1]{fontenc}
\usepackage[USenglish]{babel}
\usepackage{lmodern} % latin modern fonts
\usepackage{graphicx}
\usepackage{subcaption} \DeclareCaptionSubType[roman]{figure} % define the subtype labeling of the subcaptionbox figures
\usepackage{enumitem} \setlist[enumerate,1]{label={(\alph*)}} % roemische Zahlen auf der ersten Stufe, for compactenum use (or add here) options [topsep=0pt,itemsep=-1ex,partopsep=1ex,parsep=1ex], for compactitemize add \setitemize{noitemsep,topsep=0pt,parsep=0pt,partopsep=0pt}
\usepackage{microtype}
\usepackage{array}
\usepackage{cite}

\usepackage{tikz}
\usepackage{tkz-berge}
\usetikzlibrary{decorations.pathreplacing}
\usetikzlibrary{topaths,calc,graphs, math}
\usetikzlibrary{positioning}
\usepackage{pgfmath}
\usetikzlibrary{backgrounds}
\tikzset{
	position/.style args={#1:#2 from #3}{at=(#3), shift=(#1:#2)},
	nnode/.style={circle, draw=yellow!30!white, fill=yellow!30!white,  thin, inner sep=0pt, minimum size=10pt},
	nsnode/.style={rectangle, draw=yellow!30!white, fill=yellow!30!white,  thin, inner sep=0pt, minimum size=10pt},
	nline/.style={line width=5pt, yellow!30!white},
	rnode/.style={circle, draw=black, fill=white,  thin, inner sep=0pt, minimum size=4pt},
	fnode/.style={circle, draw=black, fill=black,  thin, inner sep=0pt, minimum size=8pt},
	cnode/.style={circle, draw=white, fill=white,  thin, inner sep=0pt, minimum size=0pt},
	snode/.style={circle, draw=black, fill=black,  thin, inner sep=0pt, minimum size=4pt},
	line/.style = { draw, thick, -{stealth} },
	dline/.style = { draw, thick, -{stealth}, dotted },
	nonarrow/.style = { draw, thick},
}

\graphicspath{{./figures/}{../figures/}}

\definecolor{hellblau}{rgb}{0.2,0.4,1} % eigene Farben definieren
\definecolor{dunkelblau}{rgb}{0,0,0.8}
\definecolor{dunkelgruen}{rgb}{0,0.5,0}
\definecolor{green}{rgb}{0,0.6,0}
\definecolor{fillblack}{rgb}{0.95,0.95,0.95}
\usepackage[
pdftex,
colorlinks,
linkcolor=dunkelblau,
urlcolor=dunkelblau,
citecolor=dunkelgruen,
bookmarks=true,
linktocpage=true,
pdfauthor={},
pdfsubject={}%,
]{hyperref}
\usepackage{pdfpages} % um Appendix/PDFs einzubinden

\usepackage[textsize=scriptsize]{todonotes}
\usepackage{lineno} % adds line numbers for review purposes
\usepackage{amsmath} % may be replaced by the superset \package{mathtools}
\usepackage{amsthm}
\usepackage{amssymb}
\usepackage{amsfonts} % Math-Befehle wie natürliche, reele Zahlen einbinden
\usepackage{mathtools} % mathclap,substack
\theoremstyle{plain} % kursiv
\newtheorem{satz}{Satz}[] %section
\newtheorem{theorem}[satz]{Theorem}
\newtheorem*{theorem*}{Theorem}
\newtheorem{lemma}[satz]{Lemma}
\newtheorem*{lemma*}{Lemma}
\newtheorem{corollary}[satz]{Corollary}

\theoremstyle{remark} % normal
\theoremstyle{definition} % normal mit fettem Titel
\newtheorem{definition}[satz]{Definition}

\newtheorem*{conjecture*}{Conjecture}

\newcommand{\red}{{\color{red}1}}
\newcommand{\green}{{\color{dunkelgruen}2}}
\newcommand{\blue}{{\color{blue}3}}

\title{Trees and co-trees in planar 3-connected graphs\\ An easier proof via Schnyder woods}
\author{Christian Ortlieb\\Institute of Computer Science\\University of Rostock\thanks{This research is supported by the grant SCHM 3186/2-1 (401348462) from the Deutsche Forschungsgemeinschaft (DFG, German Research Foundation).}
\and Jens M. Schmidt\\Institute of Computer Science\\University of Rostock\footnotemark[1]}

\begin{document}
\maketitle
%\linenumbers
\pgfdeclarelayer{background}
\pgfsetlayers{background,main}

	\begin{abstract}
	Let $G$ be a 3-connected planar graph. Define the \emph{co-tree} of a spanning
	tree $T$ of $G$ as the graph induced by the dual edges of $E(G)-E(T)$. The
	well-known cut-cycle duality implies that the co-tree is itself a tree.
	Let a \emph{$k$-tree} be a spanning tree with maximum degree $k$.
	
	In 1970, Grünbaum conjectured that every 3-connected planar graph
	contains a 3-tree whose co-tree is also a 3-tree. In 2014, Biedl \cite{B14}
	showed that every such graph contains a 5-tree whose co-tree is a
	5-tree. In this paper, we present an easier proof of Biedl's result
	using Schnyder woods.
\end{abstract}

\section{Introduction}
A fundamental theorem shown~1966 by Barnette~\cite{Barnette1966} states that every 3-connected planar graph contains a spanning 3-tree. The aforementioned conjecture of Grünbaum~\cite{Gruenbaum1970} asks to strengthen Barnette's theorem by bounding simultaneously the maximum degree of the co-tree to three. There has not been any progress on this specific question until Biedl~\cite{B14} proved~2014 the existence of a 5-tree whose co-tree is also a 5-tree.

While Biedl's proof uses canonical orderings, we use Schnyder woods to reproof her result. Schnyder woods evolved into a powerful tool during the last decades for graph drawing and the structure of 3-connected planar graphs~\cite{Aerts2015,Alam2015,F04GGAA,Felsner2001,Schnyder1990}. Although this tool has already become standard machinery, it still takes some effort to give the necessary definitions of Schnyder woods and their relatives such as ordered path partitions; we will give these in Section~\ref{sec_preliminaries}. However, using this machinery does not only allow for an easier and more standard notation, it also shortens and simplifies the proof itself a lot (in this paper, the proof starts in Section~\ref{sec_alt_proof}).

The roadmap of this proof is to define a candidate graph using Schnyder woods followed by alternative proofs of Lemmas~3, 4 and 5 of Biedl's paper~\cite{B14} using structural results on Schnyder woods. Finally, we use our candidate graph and apply Theorem~3 of~\cite{B14}, which effectively bypasses the more sophisticated parts of Biedl's proof.

\section{Preliminaries}\label{sec_preliminaries}
We use standard graph notation. Let $G$ be a graph that is simple, planar, 3-connected and comes with a fixed embedding into the plane (we say that $G$ is \emph{plane}). %We use the definition of Schnyder woods as given by Felsner~\cite{F04GGAA}.
A \emph{half-edge} is an arc that starts at a vertex but has no defined end vertex. A \emph{suspension} $G^\sigma$ of $G$ is obtained by choosing three different vertices $r_\red$, $r_\green$ and $r_\blue$ that appear in clockwise order on the outer face of $G$ and by adding a half-edge adjacent to each of those vertices that reaches into the outer face of $G$. The special vertices $r_\red$, $r_\green$ and $r_\blue$ are called \emph{roots}.

\begin{definition}\label{def:Schnyderwood}
Given a suspension $G^\sigma$, a \emph{Schnyder wood} rooted at $r_\red$, $r_\green$ and $r_\blue$ is an orientation and coloring of the edges of $G^\sigma$ (including half-edges) with colors \red, \green\ and \blue\ satisfying the following conditions (see Figures~\ref{fig:SchnyderWoodCondition} and \ref{fig_example_SW}).
\begin{enumerate}
	\item Every edge $e$ is oriented in one direction (we say $e$ is \emph{unidirected}) or in two opposite directions (we say $e$ is \emph{bidirected}). Every direction of an edge is colored with one of the three colors \red, \green, \blue\ (we say an edge is $i$-\emph{colored} if one of its directions has $i$) such that the two colors $i$ and $j$ of every bidirected edge are distinct (we call such an edge $i$-$j$-\emph{colored}). Similarly, a unidirected edge whose direction has color $i$ is called $i$-\emph{colored}. Throughout the paper, we assume modular arithmetic on the colors \red, \green, \blue\ in such a way that $i+1$ and $i-1$ for a color $i$ are defined as $(i \mod 3) +1$ and $(i+1 \mod 3) +1$. For a vertex $v$, a uni- or bidirected edge is \emph{incoming} ($i$-colored) \emph{in} $v$ if it has a direction (of color $i$) that is directed toward $v$, and \emph{outgoing} ($i$-colored) \emph{of} $v$ if it has a direction (of color $i$) that is directed away from $v$.\label{def:Schnyderwood1}
	\item For every color~$i$, the half-edge at $r_i$ is unidirected, outgoing and $i$-colored.\label{def:Schnyderwood2}
	\item Every vertex $v$ has exactly one outgoing edge of every color. The outgoing \red-, \green-, \blue-colored edges $e_\red,e_\green,e_\blue$ of $v$ occur in clockwise order around $v$. For every color~$i$, every incoming $i$-colored edge of $v$ is contained in the clockwise sector around $v$ from $e_{i+1}$ to $e_{i-1}$ (see Figure~\ref{fig:SchnyderWoodCondition}).\label{def:Schnyderwood3}
	\item No inner face boundary contains a directed cycle (disregarding possible opposite edge directions) in one color.\label{def:Schnyderwood4}
\end{enumerate}
\end{definition}

Felsner~\cite{Felsner2001} showed that every 3-connected plane graph has a Schnyder wood. Throughout the paper we use red, green and blue synonymously for colors \red, \green\ and \blue, respectively. Denote by $T_i$ the directed graph induced by the directed edges that have color $i$. $T_\red$, $T_\green$ and $T_\blue$ are called the \emph{trees} of the Schnyder wood.

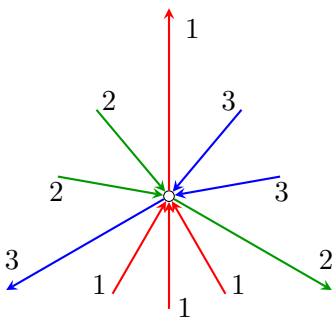
\begin{figure}[!htb]
	\centering
	\begin{tikzpicture}[scale = 0.25]
			\node[rnode] (0) at (0,0) {};
			
			\draw [line, red] (0) to (90:10);
			\draw [line, blue] (0) to (210:10);
			\draw [line, green] (0) to (330:10);
			
			\draw [line, red] (240:6) to (0);
			\draw [line, red] (270:6) to (0);
			\draw [line, red] (300:6) to (0);
			
			\draw [line, blue] (10:6) to (0);
			\draw [line, blue] (50:6) to (0);
			
			\draw [line, green] (130:6) to (0);
			\draw [line, green] (170:6) to (0);
			
			\node[] (1) at (82:9) {1};
			\node[] (2) at (338:9) {2};
			\node[] (3) at (202:9) {3};
			
			\node[] (11) at (232:6) {1};
			\node[] (12) at (278:6) {1};
			\node[] (13) at (308:6) {1};
			
			\node[] (21) at (122:6) {2};
			\node[] (23) at (178:6) {2};
			
			\node[] (31) at (2:6) {3};
			\node[] (32) at (58:6) {3};
			
		\end{tikzpicture}
	\caption{Example for Condition~\ref{def:Schnyderwood}\ref{def:Schnyderwood3} at a vertex in a Schnyder wood. The incoming edges in color $i$ are in the clockwise sector between the outgoing edge in color $i+1$ and the outgoing edge in color $i-1$.}
	\label{fig:SchnyderWoodCondition}
\end{figure}

\begin{figure}[!htb]
	\centering
	\begin{tikzpicture}[scale = 0.6,]
			\node[rnode, label =45:{$r_\red$}] (r1) at (0,10) {};
			\node[rnode, label =below:{$r_\green$}] (r2) at (6,0) {};
			\node[rnode, label =below:{$r_\blue$}] (r3) at (-6,0) {};
			\node[rnode] (1) at (-2,0) {};
			\node[rnode] (2) at (2,0) {};
			\node[rnode] (3) at (-1.5,2) {};
			\node[rnode] (4) at (1.5,2) {};
			\node[rnode] (5) at (-2.2,3.3) {};
			\node[rnode] (6) at (-0.5,3.3) {};
			\node[rnode] (7) at (2.5,3) {};
			\node[rnode] (8) at (0.5,5) {};
			\node[rnode] (9) at (0,8) {};
			\foreach \x/\y in {r3/1, 1/2, 2/r2, 3/4}{
					\draw[line, green] (\x) to ($(\x) !0.5! (\y)$);
					\draw[line, blue] (\y) to ($(\x) !0.5! (\y)$);
				}
			\foreach \x/\y in {r3/r1, 5/9, 1/3, 6/8}{
					\draw[line, red] (\x) to ($(\x) !0.5! (\y)$);
					\draw[line, blue] (\y) to ($(\x) !0.5! (\y)$);
				}
			\foreach \x/\y in {r1/r2, 9/8, 6/4, 5/3, 4/2, 8/7}{
					\draw[line, green] (\x) to ($(\x) !0.5! (\y)$);
					\draw[line, red] (\y) to ($(\x) !0.5! (\y)$);
				}
			\foreach \x/\y in {5/r3, 6/5, 7/4}{
					\draw[line, blue] (\x) to (\y);
				}
			
			\draw[line, green] (7) to (r2);
			\draw[line, red] (9) to (r1);
			
			\draw[line, red] (r1) to (0,10.8);
			\draw[line, green] (r2) to (6.6,-0.4);
			\draw[line, blue] (r3) to (-6.6,-0.4);

		\end{tikzpicture}
	\caption{A Schnyder wood of the suspension of a 3-connected planar graph.}
	\label{fig_example_SW}
\end{figure}
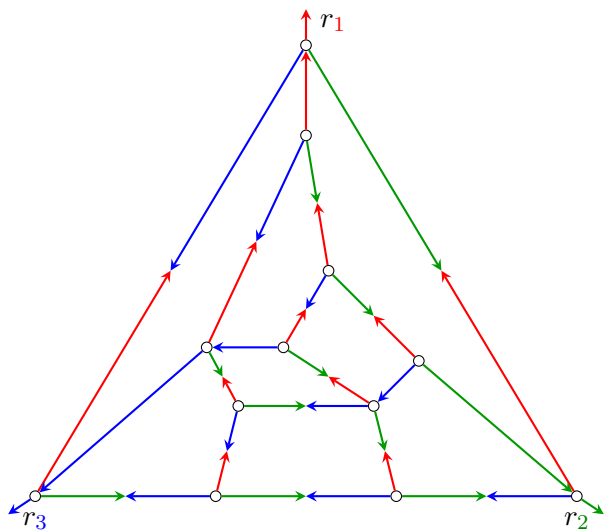

	\paragraph{Dual Schnyder Woods.}
Let $G$ be a 3-connected plane graph. Any Schnyder wood of $G^\sigma$ induces a Schnyder wood of a slightly modified planar dual of $G^\sigma$ in the following way~\cite{BTV99,F04LSFPG} (see~\cite[p.~30]{Kant1996} for an earlier variant of this result given without proof). As common for plane duality, we will use the plane dual operator $^*$ to switch between primal and dual objects, also on sets of objects.

Extend the three half-edges of $G^\sigma$ to non-crossing infinite rays and consider the planar dual of this plane graph. Since the infinite rays partition the outer face $f$ of $G$ into three parts, this dual contains a triangle with vertices $b_\red$, $b_\green$ and $b_\blue$ instead of the outer face vertex $f^*$ such that $b^*_i$ is not incident to $r_i$ for every $i$ (see Figure~\ref{fig:completion}). Let the \emph{suspended dual} $G^{\sigma^*}$ of $G$ be the graph obtained from this dual by adding at each vertex of $\{b_\red,b_\green,b_\blue\}$ a half-edge pointing into the outer face.

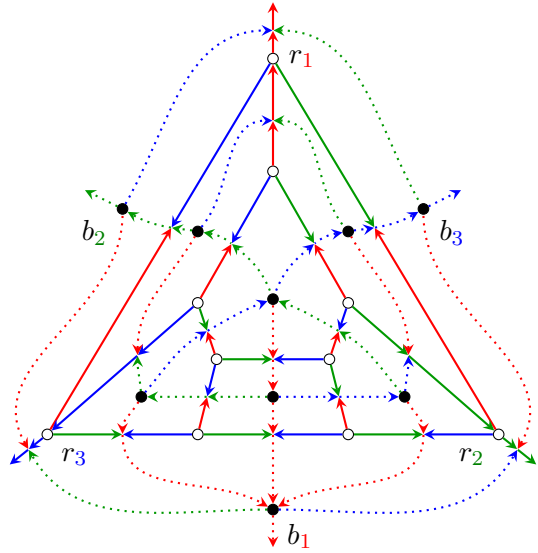
\begin{figure}[!htb]
	\centering
	\begin{tikzpicture}[scale = 0.5]
		\node[rnode, label =right:{$r_\red$}] (r1) at (0,10) {};
		\node[rnode, label =-135:{$r_\green$}] (r2) at (6,0) {};
		\node[rnode, label =-45:{$r_\blue$}] (r3) at (-6,0) {};
		\node[rnode] (1) at (-2,0) {};
		\node[rnode] (2) at (2,0) {};
		\node[rnode] (3) at (-1.5,2) {};
		\node[rnode] (4) at (1.5,2) {};
		\node[rnode] (5) at (-2,3.5) {};
		\node[rnode] (7) at (2,3.5) {};
		\node[rnode] (9) at (0,7) {};
		
		\node[snode, label =-45:{$b_\red$}] (b1) at(0,-2) {};
		\node[snode, label =190:{$b_\green$}] (b2) at(-4,6) {};
		\node[snode, label =-10:{$b_\blue$}] (b3) at(4,6) {};
		\node[snode] (d1) at(-3.5,1) {};
		\node[snode] (d2) at(0,1) {};
		\node[snode] (d3) at(3.5,1) {};
		\node[snode] (d4) at(-2,5.4) {};
		\node[snode] (d5) at(0,3.6) {};
		\node[snode] (d6) at(2,5.4) {};
		
		\foreach \z/\x/\y in {1/r3/1, 2/1/2, 3/2/r2, 4/1/3, 5/2/4, 8/3/5, 9/4/7, m/3/4}{
			\node[cnode] (c\z) at ($(\x) !0.5! (\y)$) {};
		}
		\foreach \z/\x/\y in {6/r3/5, 7/r2/7}{
			\node[cnode] (c\z) at ($(\x) !0.6! (\y)$) {};
		}
		\foreach \z/\x/\y in {12/r3/r1, 13/r2/r1, 10/9/5, 11/9/7, 14/r1/9}{
			\node[cnode] (c\z) at ($(\x) !0.55! (\y)$) {};
		}

		\foreach \x/\y in {c14/r1, 9/c14, r2/c13, r3/c12, 5/c10, 7/c11, 1/c4, 2/c5, 3/c8, 4/c9}{
			\draw[line, red] (\x) to (\y);
		}
		\foreach \x/\y in {c6/r3, 5/c6, r1/c12, 1/c1, 2/c2, r2/c3, 9/c10, 7/c9, 3/c4, 4/cm}{
			\draw[line, blue] (\x) to (\y);
		}
		\foreach \x/\y in {c7/r2, 7/c7, r1/c13, r3/c1, 1/c2, 2/c3, 9/c11, 5/c8, 3/cm, 4/c5}{
			\draw[line, green] (\x) to (\y);
		}

		\draw [dline, green] (c4) to (d1);
		\draw [dline, green] (d2) to (c4);
		\draw [dline, blue] (d2) to (c5);
		\draw [dline, blue] (c5) to (d3);
		\draw [dline, red] (d1) to [in=90, out=-120] (c1);
		\draw [dline, red] (c1) to [in=150, out=-90] (b1);
		\draw [dline, red] (d2) to (c2);
		\draw [dline, red] (c2) to (b1);
		\draw [dline, red] (d3) to [in=90, out=-60] (c3);
		\draw [dline, red] (c3) to [in=30, out=-90] (b1);
		\draw [dline, green] (d1) to [in=-90, out=95] (c6);
		\draw [dline, red, {stealth}-] (c6) to [in=-120, out=90] (d4);
		\draw [dline, blue] (d3) to [in=-90, out=85] (c7);
		\draw [dline, red, {stealth}-] (c7) to [in=-60, out=90] (d6);
		\draw [dline, blue] (d1) to [in=-150, out=50] (c8);
		\draw [dline, blue] (c8) to [in=-160, out=30] (d5);
		\draw [dline, red, {stealth}-] (d2) to (cm);
		\draw [dline, red, {stealth}-] (cm) to (d5);
		\draw [dline, green] (d3) to [in=-30, out=130] (c9);
		\draw [dline, green] (c9) to [in=-20, out=150] (d5);
		\draw [dline, green, {stealth}-] (d4) to [in=140, out=-10] (c10);
		\draw [dline, green, {stealth}-] (c10) to [in=110, out=-40] (d5);
		\draw [dline, blue, {stealth}-] (d6) to [in=40, out=-170] (c11);
		\draw [dline, blue, {stealth}-] (c11) to [in=70, out=-140] (d5);
		\draw [dline, green] (d4) to [in=-10, out=170] (c12);
		\draw [dline, green] (c12) to [in=-30, out=170] (b2);
		\draw [dline, blue] (d6) to [in=-170, out=10] (c13);
		\draw [dline, blue] (c13) to [in=-150, out=10] (b3);
		\draw [dline, blue] (d4) to [in=180, out=50] (c14);
		\draw [dline, green, {stealth}-] (c14) to [in=130, out=0] (d6);
		
		\node[cnode] (cr3) at ($(r3) !0.5! (-7,-0.8)$) {};
		\node[cnode] (cr1) at ($(r1) !0.5! (0,11.5)$) {};
		\node[cnode] (cr2) at ($(r2) !0.5! (7,-0.8)$) {};
		
		\draw [dline, green] (b1) to [in=-60, out=180] (cr3);
		\draw [dline, red, {stealth}-] (cr3) to [in=-90, out=120] (b2);
		\draw [dline, blue] (b1) to [in=-120, out=0] (cr2);
		\draw [dline, red, {stealth}-] (cr2) to [in=-90, out=60] (b3);		
		\draw [dline, blue] (b2) to [in=-180, out=60] (cr1);
		\draw [dline, green, {stealth}-] (cr1) to [in=120, out=0] (b3);		
		
		\draw[line, red] (r1) to (cr1);
		\draw[line, green] (r2) to (cr2);
		\draw[line, blue] (r3) to (cr3);

		\draw[dline, green] (b2) to (-5,6.5);
		\draw[dline, red] (b1) to  (0,-3);
		
		\draw[line, red] (cr1) to (0,11.5);
		\draw[line, green] (cr2) to  (7,-0.8);
		\draw[line, blue] (cr3) to  (-7,-0.8);
		\draw[dline, blue] (b3) to (5,6.5);

%		\begin{pgfonlayer}{background}
%			\draw[nline] (d1.center) to [in=-150, out=50] (c8.center) to [in=-160, out=30] (d5.center) to (cm.center) to (d2.center) to (c4.center) -- cycle;
%		\end{pgfonlayer}
		
	\end{tikzpicture}
	\caption{The completion of $G$ obtained by superimposing $G^\sigma$ and its suspended dual $G^{\sigma^*}$ (the latter depicted with dotted edges).}
	\label{fig:completion}
\end{figure}

Consider the superposition of $G^\sigma$ and its suspended dual $G^{\sigma^*}$ such that exactly the primal dual pairs of edges cross (here, for every $\red \leq i \leq \blue$, the half-edge at $r_i$ crosses the dual edge $b_{i-1}b_{i+1}$).

\begin{definition}\label{def:schnyderdual}
	For any Schnyder wood $S$ of $G^\sigma$, define the orientation and coloring $S^*$ of the suspended dual $G^{\sigma^*}$ as follows (see Figure~\ref{fig:completion}):
	\begin{enumerate}
		\item For every unidirected $(i-1)$-colored edge or half-edge $e$ of $G^\sigma$, color $e^*$ with the two colors $i$ and $i+1$ such that $e$ points to the right of the $i$-colored direction.
		\item Vice versa, for every $i$-$(i+1)$-colored edge $e$ of $G^\sigma$, $(i-1)$-color $e^*$ unidirected such that $e^*$ points to the right of the $i$-colored direction.
		\item For every color $i$, make the half-edge at $b_i$ unidirected, outgoing and $i$-colored.\label{def:schnyderdual3}
	\end{enumerate}
\end{definition}

The following lemma states that $S^*$ is indeed a Schnyder wood of the suspended dual. By Definition~\ref{def:schnyderdual}\ref{def:schnyderdual3}, the vertices $b_\red$, $b_\green$ and $b_\blue$ are the roots of $S^*$.

\begin{lemma}[{\cite{Kant1992}\cite[Prop.~3]{F04LSFPG}}]\label{lem:schnyderdual}
	For every Schnyder wood $S$ of $G^\sigma$, $S^*$ is a Schnyder wood of $G^{\sigma^*}$.
\end{lemma}

Since $S^{*^*} = S$, Lemma~\ref{lem:schnyderdual} gives a bijection between the Schnyder woods of $G^\sigma$ and the ones of $G^{\sigma^*}$.
Let the \emph{completion} $\widetilde{G}$ of $G$ be the plane graph obtained from the superposition of $G^\sigma$ and $G^{\sigma^*}$ by subdividing each pair of crossing (half-)edges with a new vertex, which we call a \emph{crossing vertex} (see Figure~\ref{fig:completion}). The completion has six half-edges pointing into its outer face.

Any Schnyder wood $S$ of $G^\sigma$ implies the following natural orientation and coloring $\widetilde{G}_S$ of its completion $\widetilde{G}$. Let $vw \in E(G^\sigma) \cup E(G^{\sigma^*})$, let $z$ be the crossing vertex of $G^\sigma$ that subdivides $vw$ and consider the coloring of $vw$ in either $S$ or $S^*$. If $vw$ is outgoing of $v$ and $i$-colored, we direct $vz \in E(\widetilde{G})$ toward $z$ and $i$-color it (and do the same for all other vertices than $v$). In the remaining case that $vw$ is unidirected, incoming in $v$ and $i$-colored, we direct $zv \in E(\widetilde{G})$ toward $v$ and $i$-color it. The three half-edges of $G^{\sigma^*}$ inherit the orientation and coloring of $S^*$ for $\widetilde{G}_S$.
By Definition~\ref{def:schnyderdual}, the construction of $\widetilde{G}_S$ implies immediately the following corollary.

\begin{corollary}\label{cor:crossingvertex}
	Every crossing vertex of $\widetilde{G}_S$ has one outgoing edge and three incoming edges and the latter are colored \red, \green\ and \blue\ in counterclockwise direction.
\end{corollary}

\paragraph{Ordered Path Partitions.}

\begin{definition}\label{def:opp}
	Let $G$ be a 3-connected plane graph with vertices $r_1$, $r_2$ and $r_3$ on the boundary of the outer face in this clockwise order. For any $j \in \{1,2,3\}$, an \emph{ordered path partition} $\mathcal{P} = (P_0, \ldots, P_s)$ of $G$ \emph{with base-pair} $(r_j,r_{j+1})$ is an ordered partition of $V(G)$ into the vertex sets of induced paths (therefore often referred to as paths) such that the following holds for every $i \in \{0, \ldots, s-1\}$, where $V_i := \bigcup_{q=0}^i V(P_q)$ and the \emph{contour} $C_i$ is the clockwise walk from $r_{j+1}$ to $r_j$ on the outer face of $G[V_i]$. 
	\begin{enumerate}
		\item $P_0$ consists of the vertices of the clockwise path from $r_j$ to $r_{j+1}$ on the outer face boundary, and $P_s = \{r_{j+2}\}$.\label{def:opp_init}
		\item Each vertex in $P_i$ has a neighbor in $V(G) \setminus V_i$.\label{def:opp_at_least_one_neighbor}
		\item $C_i$ is a path.\label{def:opp_2connected} % inmplied by the following weaker condition of Badent et al.: the outer face boundary of that graph is a cycle
		\item Each vertex in $C_i$ has at most one neighbor in $P_{i+1}$.\label{def:opp_at_most_one_neighbor}
	\end{enumerate}
\end{definition}

The following lemma describes a connection between Schnyder woods and ordered path partitions. Its proof was first given by Badent et al.~\cite[Theorem~5]{BBC11}, which turned out however to be incomplete, then corrected by Alam et al.~\cite[Lemma~1]{Alam2015}, which however outsourced the proof into the extended abstract~\cite[arXiv version, Section~2.2]{Alam2015}.

\begin{lemma}[{Alam et al.~\cite[arXiv version, Section~2.2]{Alam2015}}]
	Let $S$ be a Schnyder wood of the suspension of a 3-connected plane graph $G$. For every $i \in \{\red,\green,\blue\}$ there exists an ordered path partition $\mathcal{P}^{i,i+1}$ with base pair $(r_i, r_{i+1})$ such that 
	\begin{enumerate}
		\item the paths of $\mathcal{P}^{i,i+1}$ are formed by the maximal $i$-$(i+1)$-colored paths,
		\item the order of $\mathcal{P}^{i,i+1}$ is a linear extension of the partial order defined by $T^{-1}_{i} \cup T^{-1}_{i+1} \cup T_{i+2}$.\label{lem_compatible_consistency}
	\end{enumerate}
	Call such an ordered path partition \emph{compatible}.
	\label{lem_compatible}
\end{lemma}
%
%Badent et al.~\cite[Theorem~5]{BBC11}, Alam et al~\cite[Lemma~4]{Alam2015a}, \cite[Lemma~1]{Alam2015}
%
%As proven in~\cite[Theorem~5]{BBC11} and~\cite[Lemma~4]{Alam2015a}, the vertex sets of the inclusion-wise maximal $j$-$(j+1)$-colored paths of any Schnyder wood $S$ of a suspension $G^\sigma$, $j \in \{1,2,3\}$ form an ordered path partition of $G$ with base pair $(r_j,r_{j+1})$; we call this special ordered path partition \emph{compatible} with $S$ and denote it by $\mathcal{P}^{j,j+1}$. 
%
%For an ordered path partition $\mathcal{P}^{j,j+1}$ of $G^\sigma$, every $P_k \in \mathcal{P}^{j,j+1}$ induces a path in $G$ by Definitions~\ref{def:Schnyderwood}, \ref{def:opp}\ref{def:opp_at_least_one_neighbor} and~\ref{def:opp}\ref{def:opp_at_most_one_neighbor}; we therefore call the equivalence classes of an ordered path partition \emph{paths}. By~\ref{def:opp}\ref{def:opp_init} and~\ref{def:opp}\ref{def:opp_at_least_one_neighbor}, $G$ contains for every $k$ and every vertex $v \in P_k$ a path from $v$ to $r_{j+2}$ that intersects $V_k$ only in $v$. Since $G$ is plane, this implies that every path $P_k$ is embedded into the outer face of $G[V_{k-1}]$ for every $1 \leq k \leq s$.

\section{Our alternative proof}\label{sec_alt_proof}
We give an easier proof of Biedl's theorem using Schnyder woods and essentially use them to define the candidate graph $H(G)$ and show that $H(G)$ meets the statements of Lemma~3, 4 and 5 of Biedl's proof in~\cite{B14}. This way, we can use some standard properties of Schnyder woods, which makes the proof easier to understand.

Observe that the difference of $G^{\sigma^*}$ and $G^*$ is marginal: while $G^*$ contains the vertex that corresponds to the dual of the outer face of $G$, $G^{\sigma^*}$ contains instead three vertices that correspond to the dual vertices of the three unbounded regions of $G^\sigma$. We freely switch from $G^{\sigma^*}$ to $G^*$ by identifying the three roots of the suspended dual.

\begin{definition}
	Let $G$ be a 3-connected plane graph, $S$ be a Schnyder wood of $G^\sigma$ and $\mathcal{P}^{\green,\blue} = (P_0, \ldots , P_s)$ be the compatible ordered path partition formed by the maximal \green-\blue-colored paths. For a path $P_i$ define the \emph{parent path} to be a path $P_j$ with $j < i$ maximal such that there is an edge joining $P_i$ and $P_j$. Define the \emph{parent edge} of $P_i$ to be an edge joining it to its parent path. If there is more than one edge joining $P_i$ to its parent path just choose one. Observe that, by Lemma~\ref{lem_compatible}\ref{lem_compatible_consistency}, the parent edge of a path $P_i$ is incoming \red-colored, outgoing \green-colored or outgoing \blue-colored at a vertex of $P_i$.
	\label{def_parent_edge}
\end{definition}

\begin{definition}
	For a 3-connected plane graph $G$ and a fixed Schnyder wood, define the subgraph $H(G)$ as follows. Let $V(H(G)) := V(G)$ and let an edge $e \in E(G)$ be an edge of $H(G)$ if 
	\begin{enumerate}
		\item [(H1)] $e$ is \green-\blue-colored,
		\item [(H2)] $e$ is the first incoming \blue-colored edge at one of its endpoints in clockwise direction,
		\item [(H3)] $e$ is the last incoming \green-colored edge at one of its endpoints in clockwise direction,
		\item [(H4)] or $e$ is a parent edge and colored with color \red. 
	\end{enumerate}
	With a little abuse of notation, define $H^\circ(G^*)$ as the simple graph that is obtained from $H(G^{\sigma^*})$ by identifying the three root vertices of $G^{\sigma^*}$. Thus $H^\circ(G^*)$ is a subgraph of $G^*$.
	
	Observe that $H(G)$ contains all bidirected edges of $G$. By~(H1), $H(G)$ contains all \green-\blue-colored edges. By (H3), every \red-\green-colored edge is in $E(H(G))$. By (H2), every edge with colors~\red\ and \blue\ is in $E(H(G))$. The same arguments show that $H^\circ(G^*)$ contains all bidirected edges of $G^*$ that remain after the identification of the roots of $G^{\sigma^*}$. 
\end{definition}
\begin{lemma}
	$H(G)$ and $H^\circ(G^*)$ have maximum degree at most 5.
	\label{lem_deg_leq_5}
\end{lemma}
\begin{proof}
	First consider $H(G)$. Let $v$ be a vertex of $H(G)$. The three outgoing edges of $v$ are possibly in $H(G)$. Also there is at most one parent edge that is incoming \red-colored at $v$, and rules (H2) and (H3) are responsible for at most one edge each. So $\deg_{H(G)}(v) \leq 6$ for all vertices $v \in V(H)$. Assume, for the sake of contradiction, that there is a vertex $v$ with $\deg_{H(G)}(v) = 6$. Consider the incoming edges $vv_r$ and $vv_l$ of $v$ given by rule (H2) and (H3), respectively, and the outgoing \red-colored edge $vv_m$ of $v$ given by (H4). Since $\deg_{H(G)}(v) = 6$, those edges are pairwise distinct and unidirected (Figure~\ref{fig_deg_leq_5}).
	
	Let $\mathcal{P}^{\green,\blue} = (P_0, \ldots , P_s)$ be the compatible ordered path partition formed by the maximal \green-\blue-colored paths. For every vertex $v \in V(G)$ with $v \in P_t$ define $t$ to be the \emph{index} of $v$. Observe that, by Lemma~\ref{lem_compatible}\ref{lem_compatible_consistency}, for a \red-colored edge the index of its head is larger than the index of its tail (head and tail here with respect to color \red, the edge might also be bidirected.). Also, if an edge is unidirected \green- or \blue-colored, the index of its head is smaller than the index of its tail. And if an edge is \green-\blue-colored, the indices of the endpoints are equal. 
	
	Let $v \in P_j$, $v_l \in P_q$, $v_m \in P_i$ and $v_r \in P_p$. Observe that $vv_m$ is the parent edge of $P_i$. Since $vv_l$ is the last incoming \green-colored edge, $vv_m$ is outgoing \red-colored and $vv_r$ is the first incoming \blue-colored edge, they occur consecutively in this clockwise order around $v$ starting with $vv_l$. Let $f$ be the face that has $v_rv$ and $v_mv$ on its boundary (Figure~\ref{fig_deg_leq_5}). By \cite[Lemma~12]{BTV99}, there exists a path $P$ from $v_r$ to $P_i$ along the boundary of $f$ that consists of \red-\blue-colored edges and an edge which is either unidirected \red-colored, unidirected \green-colored or \red-\green-colored. Those edges are such that color \blue\ and \green\ are directed towards $v_r$ and color \red\ is directed towards $P_i$. So $P$ is non-decreasing in index. Since $vv_r$ is unidirected and ingoing \blue-colored at $v$, $p > j$ by Lemma~\ref{lem_compatible}\ref{lem_compatible_consistency}. Hence, the second to last vertex of $P$ has a higher index than $v$ and is adjacent to a vertex of $P_i$. Thus, $vv_m$ is not the parent edge of $P_i$ by Definition~\ref{def_parent_edge}, which is a contradiction, so that $v$ has degree at most 5.
	
	Consider $H^\circ(G^*)$. By the above arguments, $H(G^{\sigma^*})$ has maximum degree at most 5, so that we only need to consider the dual vertex $x$ of the outer face of $G$. As $x$ results from the identification of the three roots of $G^{\sigma^*}$, $x$ does not have outgoing edges. Hence, $x$ has maximum degree 3, which implies that $H^\circ(G^*)$ has maximum degree at most 5.
\end{proof}

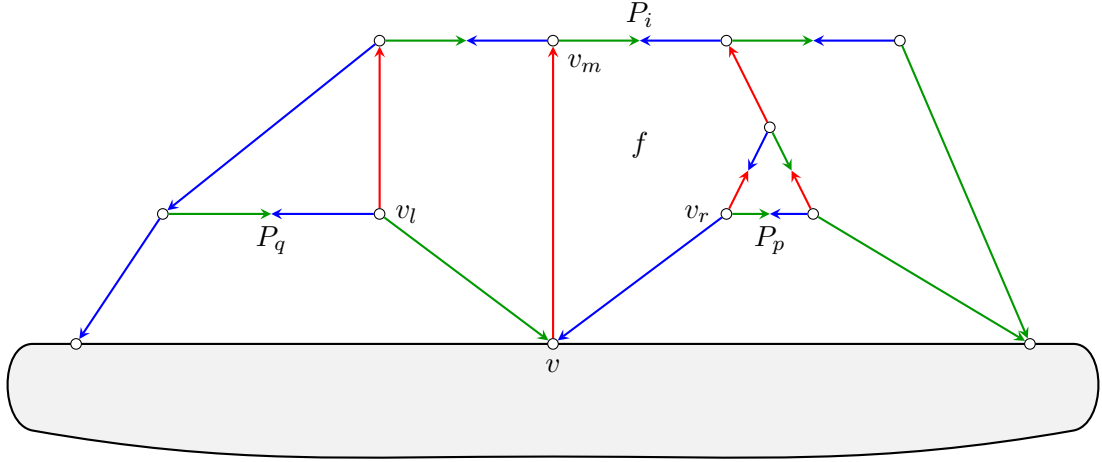
\begin{figure}[!htb]
	\centering
	\begin{tikzpicture}[scale=\textwidth/13cm]
		
		\filldraw[nonarrow, color=black, fill=fillblack] (0,0) to (12,0) to [in=10, out=0] (12,-1) to [in=0, out=-170] (6,-1.3) to [in=-10, out=180] (0,-1) to [in=180, out=170] (0,0);
		\node[] (labelnode2) at (7,3.8) {$P_i$};
		
		\node[] (labelnode2) at (8.5,1.2) {$P_p$};
		\node[] (labelnode2) at (2.75,1.2) {$P_q$};
		\node[] (labelnode2) at (7,2.3) {$f$};

		\node[rnode] (v0) at (1.5,1.5) {};
		\node[rnode] (v1) at (4,3.5) {};
		\node[rnode, label=-45:{$v_m$}] (v2) at (6,3.5) {};
		\node[rnode] (v3) at (8,3.5) {};
		\node[rnode] (v4) at (10,3.5) {};
		\node[rnode] (v5) at (11.5,0) {};
		
		\node[rnode, label=below:{$v$}] (v) at (6,0) {};
		\node[rnode, label=right:{$v_l$}] (vl) at (4,1.5) {};
		\node[rnode, label=left:{$v_r$}] (vr) at (8,1.5) {};
		\node[rnode] (vrr) at (9,1.5) {};
		\node[rnode] (vra) at (8.5,2.5) {};
		\node[rnode] (vbl) at (0.5,0) {};
		
		\foreach \x/\y in {v1/v2, v2/v3, v3/v4, vr/vrr, v0/vl}{
			\draw[line, green] (\x) to ($(\x) !0.5! (\y)$);
			\draw[line, blue] (\y) to ($(\x) !0.5! (\y)$);
		}
		\foreach \x/\y in {v/v2, vra/v3, vl/v1}{
			\draw[line, red] (\x) to (\y);
		}
		\foreach \x/\y in {vl/v, vrr/v5, v4/v5}{
			\draw[line, green] (\x) to (\y);
		}
		\foreach \x/\y in {vr/v, v1/v0, v0/vbl}{
			\draw[line, blue] (\x) to (\y);
		}
		
		\draw[line, red] (vr) to ($(vr) !0.5! (vra)$);
		\draw[line, blue] (vra) to ($(vr) !0.5! (vra)$);
		
		\draw[line, red] (vrr) to ($(vrr) !0.5! (vra)$);
		\draw[line, green] (vra) to ($(vrr) !0.5! (vra)$);

		%		\begin{pgfonlayer}{background}
			%			\draw[nline] (7,2) to (y.center) to ($(x) !0.5! (v3)$) to (v3.center) -- cycle;
			%		\end{pgfonlayer}
		
	\end{tikzpicture}
	\caption{Illustration of the proof of Lemma~\ref{lem_deg_leq_5}.}
	\label{fig_deg_leq_5}
\end{figure}

\begin{lemma}\label{lem_dual_edge}
	Let $e_1 \in E(G)\setminus E(H(G))$. Then the dual edge $e_1^*$ is in $E(H^\circ(G^*))$. Vice versa, for every edge $e_2 \in E(G^*) \setminus E(H^\circ(G^*))$, we have $e_2^* \in E(H(G))$.
\end{lemma}
\begin{proof}
	Since all bidirected edges of $G$ are in $H(G)$, $e_1$ is unidirected. By Corollary~\ref{cor:crossingvertex}, $e_1^*$ is bidirected and hence $e_1^* \in E(H^\circ(G^*))$. The same argument also works for $e_2$.
\end{proof}

\begin{lemma}\label{lem:cycle}
	Let $C$ be a cycle in $H(G)$. Then there exists an edge $e_1 \in C$ such that $e_1^* \in E(H^\circ(G^*))$. Vice versa, for every cycle $C'$ in $H^\circ(G^*)$, there is an edge $e_2 \in C'$ such that $e_2^* \in E(H(G))$.
\end{lemma}
\begin{proof}
	Let $x$ be the dual vertex of the outer face of $G$. If the cycle $C'$ of $H^\circ(G^*)$ contains $x$, then, there is a unidirected edge $e$ on $C'$, as $x$ is incident to only unidirected edges. Since $e^*$ is bidirected, $e^* \in E(H(G))$. Hence, we may assume that $C'$ contains only dual vertices of internal faces. The same arguments apply to $C$ and $C'$, so that we only need to consider $C$ in the following.
	
	If there is an unidirected edge $e \in C$, then $e^*$ is bidirected by Corollary~\ref{cor:crossingvertex} and hence $e^* \in E(H^\circ(G^*))$. So let $C$ be a cycle that has only bidirected edges. Let $\mathcal{P}^{\green,\blue} = (P_0, \ldots , P_s)$ be the compatible ordered path partition formed by the maximal \green-\blue-colored paths.
	
	Let $P$ be the maximum path in $C$ such that $P \subseteq P_i$ with $i = \min \{i \mid P_i \cap C \neq \emptyset \}$; we call $P$ \emph{index minimal}.
	Let $P = (v_1, \ldots, v_k)$ such that $P$ starts at $v_1$ with an edge outgoing in color \green\ and incoming in color \blue\ and continues in counterclockwise direction (using only \green-\blue-colored edges) around $C$ ending at $v_k$. 
	
	Assume, for the sake of contradiction, that $P$ consists of only one vertex $v$. The outgoing \green-colored edge $vw$ and \blue-colored edge $vu$ are both not \green-\blue-colored. Hence, by Lemma~\ref{lem_compatible}\ref{lem_compatible_consistency}, the index of $w$ and $u$ is smaller than the index of $v$. By the index-minimality of $P$, $u$ and $w$ are not in $C$. Thus, the edges of $C$ that are incident to $v$ cannot both be bidirected, which is a contradiction. Hence, $P$ consists of at least two vertices. Let $v_0$ ($v_{k+1}$) be the clockwise (counterclockwise) neighbor of $v_1$ ($v_k$) on $C$.
	
	Since $P$ is index minimal, $v_0v_1$ is not outgoing \blue-colored at $v_1$ and $v_kv_{k+1}$ is not outgoing \green-colored at $v_k$.
	
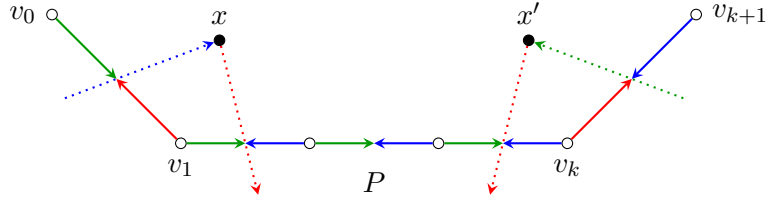
	\begin{figure}[!htb]
		\centering
		\begin{subfigure}{.8\textwidth}	
			\centering
			\begin{tikzpicture}[scale=\textwidth/7cm]		
				\node[] (labelnode2) at (2.5,-0.3) {$P$};

				\node[rnode, label=below:{$v_1$}] (v1) at (1,0) {};
				\node[rnode] (v2) at (2,0) {};
				\node[rnode] (v3) at (3,0) {};
				\node[rnode, label=below:{$v_k$}] (v4) at (4,0) {};

				\node[rnode, label=right:{$v_{k+1}$}] (v5) at (5,1) {};
				\node[rnode, label=left:{$v_0$}] (v0) at (0,1) {};

				\node[snode, label=above:{$x'$}] (dr) at (3.7,0.8) {};
				\node[snode, label=above:{$x$}] (dl) at (1.3,0.8) {};

				\foreach \x/\y in {v1/v2, v2/v3, v3/v4}{
					\draw[line, green] (\x) to ($(\x) !0.5! (\y)$);
					\draw[line, blue] (\y) to ($(\x) !0.5! (\y)$);
				}
			
				\draw[dline, blue] ($(v0) !0.5! (v1) !-0.5! (dl)$) to (dl);
				
				\draw[dline, green] ($(v4) !0.5! (v5) !-0.5! (dr)$) to (dr);

				\draw[dline, red] (dl) to ($(v1) !0.5! (v2) !-0.5! (dl)$);
				
				\draw[dline, red] (dr) to ($(v3) !0.5! (v4) !-0.5! (dr)$);
				
				\draw[line, red] (v4) to ($(v4) !0.5! (v5)$);
				\draw[line, blue] (v5) to ($(v4) !0.5! (v5)$);
				
				\draw[line, red] (v1) to ($(v1) !0.5! (v0)$);
				\draw[line, green] (v0) to ($(v1) !0.5! (v0)$);

				%		\begin{pgfonlayer}{background}
					%			\draw[nline] (7,2) to (y.center) to ($(x) !0.5! (v3)$) to (v3.center) -- cycle;
					%		\end{pgfonlayer}
				
			\end{tikzpicture}
			\caption{If there is no other incoming \blue-colored or \green-colored edge at $v_1$ or $v_k$, respectively.}
		\end{subfigure}
		
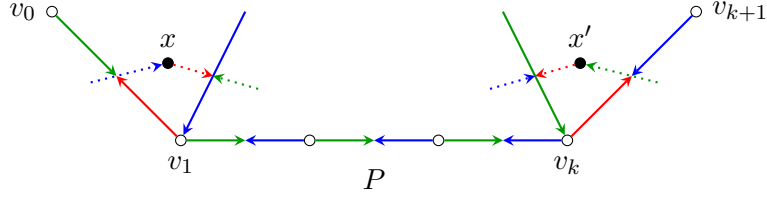
\begin{subfigure}{.8\textwidth}	
			\centering
			\begin{tikzpicture}[scale=\textwidth/7cm]
				\node[] (labelnode2) at (2.5,-0.3) {$P$};
				
				\node[rnode, label=below:{$v_1$}] (v1) at (1,0) {};
				\node[rnode] (v2) at (2,0) {};
				\node[rnode] (v3) at (3,0) {};
				\node[rnode, label=below:{$v_k$}] (v4) at (4,0) {};
				
				\node[rnode, label=right:{$v_{k+1}$}] (v5) at (5,1) {};
				\node[rnode, label=left:{$v_0$}] (v0) at (0,1) {};
				
				\node[snode, label=above:{$x'$}] (dr) at (4.1,0.6) {};
				\node[snode, label=above:{$x$}] (dl) at (0.9,0.6) {};
				
				\foreach \x/\y in {v1/v2, v2/v3, v3/v4}{
					\draw[line, green] (\x) to ($(\x) !0.5! (\y)$);
					\draw[line, blue] (\y) to ($(\x) !0.5! (\y)$);
				}
				
				\draw[dline, blue] ($(v0) !0.5! (v1) !-0.5! (dl)$) to (dl);
				
				\draw[dline, green] ($(v4) !0.5! (v5) !-0.5! (dr)$) to (dr);
				
%				\draw[dline, red] (dl) to ($(v1) !0.5! (v2) !-0.5! (dl)$);
%				
%				\draw[dline, red] (dr) to ($(v3) !0.5! (v4) !-0.5! (dr)$);
				
				\draw[line, red] (v4) to ($(v4) !0.5! (v5)$);
				\draw[line, blue] (v5) to ($(v4) !0.5! (v5)$);
				
				\draw[line, red] (v1) to ($(v1) !0.5! (v0)$);
				\draw[line, green] (v0) to ($(v1) !0.5! (v0)$);
				
				\draw[line, blue] (1.5,1) to (v1);
				\draw[line, green] (3.5,1) to (v4);
				
				\draw[dline, red] (dl) to ($(1.5,1) !0.5! (v1)$);
				\draw[dline, green] ($(1.5,1) !0.5! (v1) !-1! (dl)$) to ($(1.5,1) !0.5! (v1)$);
				
				\draw[dline, red] (dr) to ($(3.5,1) !0.5! (v4)$);
				\draw[dline, blue] ($(3.5,1) !0.5! (v4) !-1! (dr)$) to ($(3.5,1) !0.5! (v4)$);

			\end{tikzpicture}
			\caption{If there is an incoming \blue-colored or \green-colored edge at $v_1$ or $v_k$, respectively.}
		\end{subfigure}
		\caption{The two possibilities discussed in Case 1 and 2.}
		\label{fig_first_case_breaking_cycle}
	\end{figure}

	\begin{description}
		\item[Case 1:] $v_0v_1$ is \red-\green-colored. Then $(v_0v_1)^*$ is the first incoming \blue-colored edge at its head $x$ in clockwise direction and hence in $E(H(G^*))$ (see Figure~\ref{fig_first_case_breaking_cycle}).
		\item[Case 2:] $v_kv_{k+1}$ is \red-\blue-colored. Similarly to Case 1 we see that $(v_kv_{k+1})^*$ is the last incoming \green-colored edge at its head $x'$ in clockwise direction and hence in $E(H(G^*))$ (see Figure~\ref{fig_first_case_breaking_cycle}).
		
		\begin{figure}[!htb]
			\centering
			\begin{subfigure}{.8\textwidth}	
				\centering
				\begin{tikzpicture}[scale=\textwidth/7cm]		
					\node[] (labelnode2) at (2.5,-0.3) {$P$};
					
					\node[rnode, label=below:{$v_1$}] (v1) at (1,0) {};
					\node[rnode] (v2) at (2,0) {};
					\node[rnode] (v3) at (3,0) {};
					\node[rnode, label=below:{$v_k$}] (v4) at (4,0) {};
					
					\node[rnode, label=right:{$v_{k+1}$}] (v5) at (5,1) {};
					\node[rnode, label=left:{$v_0$}] (v0) at (0,1) {};
					
					\node[snode, label=right:{$x'$}] (dr) at (5,0) {};
					\node[snode, label=left:{$x$}] (dl) at (0,0) {};
					
					\foreach \x/\y in {v1/v2, v2/v3, v3/v4}{
						\draw[line, green] (\x) to ($(\x) !0.5! (\y)$);
						\draw[line, blue] (\y) to ($(\x) !0.5! (\y)$);
					}
					
					\draw[dline, green] ($(v0) !0.5! (v1) !-0.5! (dl)$) to (dl);
					
					\draw[dline, blue] ($(v4) !0.5! (v5) !-0.5! (dr)$) to (dr);
					
%					\draw[dline, red] (dl) to ($(v1) !0.5! (v2) !-0.5! (dl)$);
					
%					\draw[dline, red] (dr) to ($(v3) !0.5! (v4) !-0.5! (dr)$);
					
					\draw[line, red] (v4) to ($(v4) !0.5! (v5)$);
					\draw[line, green] (v5) to ($(v4) !0.5! (v5)$);
					
					\draw[line, red] (v1) to ($(v1) !0.5! (v0)$);
					\draw[line, blue] (v0) to ($(v1) !0.5! (v0)$);
					
					\draw[line, green] (0,-0.7) to (v1);
					
					\draw[dline, red] (dl) to ($(0,-0.7) !0.5! (v1)$);
					\draw[dline, blue] ($(0,-0.7) !0.5! (v1) !-1! (dl)$) to ($(0,-0.7) !0.5! (v1)$);
					
					\draw[line, blue] (5,-0.7) to (v4);
					
					\draw[dline, red] (dr) to ($(5,-0.7) !0.5! (v4)$);
					\draw[dline, green] ($(5,-0.7) !0.5! (v4) !-1! (dr)$) to ($(5,-0.7) !0.5! (v4)$);

					%		\begin{pgfonlayer}{background}
						%			\draw[nline] (7,2) to (y.center) to ($(x) !0.5! (v3)$) to (v3.center) -- cycle;
						%		\end{pgfonlayer}
					
				\end{tikzpicture}
				\caption{Here $v_1$ has an unidirected ingoing \green-colored edge and $v_k$ has an unidirected ingoing \blue-colored edge.}
			\end{subfigure}
			\begin{subfigure}{.8\textwidth}	
				\centering
				\begin{tikzpicture}[scale=\textwidth/7cm]		
					\node[] (labelnode2) at (2.5,-0.3) {$P$};
					
					\node[rnode, label=below:{$v_1$}] (v1) at (1,0) {};
					\node[rnode] (v2) at (2,0) {};
					\node[rnode] (v3) at (3,0) {};
					\node[rnode, label=below:{$v_k$}] (v4) at (4,0) {};
					
					\node[rnode, label=right:{$v_{k+1}$}] (v5) at (5,1) {};
					\node[rnode, label=left:{$v_0$}] (v0) at (0,1) {};
					
					\node[snode, label=right:{$x'$}] (dr) at (5,0) {};
					\node[snode, label=left:{$x$}] (dl) at (0,0) {};
					
					\foreach \x/\y in {v1/v2, v2/v3, v3/v4}{
						\draw[line, green] (\x) to ($(\x) !0.5! (\y)$);
						\draw[line, blue] (\y) to ($(\x) !0.5! (\y)$);
					}
					
					\draw[dline, green] ($(v0) !0.5! (v1) !-0.5! (dl)$) to (dl);
					
					\draw[dline, blue] ($(v4) !0.5! (v5) !-0.5! (dr)$) to (dr);
					
					%					\draw[dline, red] (dl) to ($(v1) !0.5! (v2) !-0.5! (dl)$);
					
					%					\draw[dline, red] (dr) to ($(v3) !0.5! (v4) !-0.5! (dr)$);
					
					\draw[line, red] (v4) to ($(v4) !0.5! (v5)$);
					\draw[line, green] (v5) to ($(v4) !0.5! (v5)$);
					
					\draw[line, red] (v1) to ($(v1) !0.5! (v0)$);
					\draw[line, blue] (v0) to ($(v1) !0.5! (v0)$);
					
					\draw[line, green] (0,-0.7) to ($(v1) !0.5! (0,-0.7)$);
					\draw[line, blue] (v1) to ($(v1) !0.5! (0,-0.7)$);
					
					\draw[dline, red] (dl) to ($(0,-0.7) !0.5! (v1) !-1! (dl)$);
					
					\draw[line, blue] (5,-0.7) to ($(v4) !0.5! (5,-0.7)$);
					\draw[line, green] (v4) to ($(v4) !0.5! (5,-0.7)$);
					
					\draw[dline, red] (dr) to ($(5,-0.7) !0.5! (v4) !-1! (dr)$);

					%		\begin{pgfonlayer}{background}
						%			\draw[nline] (7,2) to (y.center) to ($(x) !0.5! (v3)$) to (v3.center) -- cycle;
						%		\end{pgfonlayer}
					
				\end{tikzpicture}
				\caption{Here $v_1$ has an ingoing \green-colored edge that is \green-\blue-colored and $v_k$ has an ingoing \blue-colored edge that is \green-\blue-colored.}
			\end{subfigure}
			\caption{Two possibilities discussed in Case 3.}
			\label{fig_third_case_breaking_cycle}
		\end{figure}
	
		\item[Case 3:] If we are neither in Case 1 nor 2, then $v_0v_1$ is \red-\blue-colored and $v_kv_{k+1}$ \red-\green-colored. If the clockwise next edge $e$ around $v_k$ after $v_kv_{k+1}$ is incoming \blue-colored, then $(v_kv_{k+1})^*$ satisfies (H2) at its head $x'$ (see Figure~\ref{fig_third_case_breaking_cycle}). As $v_kv_{k+1}$ is outgoing \red-colored at $v_k$,  $e$ is either unidirected or \green-\blue-colored. 
		Similarly, observe that if the counterclockwise next edge around $v_1$ after $v_0v_1$ is incoming \green-colored, then $(v_0v_1)^*$ satisfies (H3) at its head (see Figure~\ref{fig_third_case_breaking_cycle}).

	\begin{figure}[!htb]
		\centering
		\begin{subfigure}{.8\textwidth}
			\centering
			\begin{tikzpicture}[scale=\textwidth/7cm]		
				\node[] (labelnode2) at (2.2,-0.2) {$P$};
				\node[] (labelnode3) at (1.5,-0.9) {$P'$};
				
				\node[rnode, label=left:{$v_1$}] (v1) at (1,0) {};
				\node[rnode] (v2) at (2,0) {};
				\node[rnode] (v3) at (3,0) {};
				\node[rnode, label=below:{$v_k$}] (v4) at (4,0) {};
				
				\node[rnode, label=right:{$v_{k+1}$}] (v5) at (5,1) {};
				\node[rnode, label=left:{$v_0$}] (v0) at (0,1) {};
				
				\node[snode, label=right:{$y'$}] (dr) at (5,0) {};
				\node[snode, label=left:{$y$}] (dl) at (0,0) {};
				
				\node[snode] (d1) at (0.7,0.6) {};
				\node[snode, label=above:{$x$}] (d2) at (1.5,0.6) {};
				\node[snode] (d3) at (2.5,0.6) {};
				\node[snode] (d4) at (3.5,0.6) {};
				\node[snode] (d5) at (4.3,0.6) {};
				
%				\node[snode, label=below:{$P'$}] (Pd) at (2.5,-0.9) {};

				\foreach \x/\y in {v1/v2, v2/v3, v3/v4}{
					\draw[line, green] (\x) to ($(\x) !0.5! (\y)$);
					\draw[line, blue] (\y) to ($(\x) !0.5! (\y)$);
				}
				
				\foreach \x/\y in {d2/d3, d3/d4}{
					\draw[dline, blue] (\x) to ($(\x) !0.5! (\y)$);
					\draw[dline, green] (\y) to ($(\x) !0.5! (\y)$);
				}
				
				\draw[line, red] (v2) to ($(v2) + (0,1.5)$);
				\draw[line, red] (v3) to ($(v3) + (0,1.5)$);
				
				\draw[line, blue] ($(v1) + (0,1.5)$) to (v1);
				\draw[line, green] ($(v4) + (0,1.5)$) to (v4);
				
%				\draw[dline, red] (d2) to ($(v1) !0.5! (v2)$) to [in=150, out=-90] (Pd);
%				\draw[dline, red] (d3) to (Pd);
%				\draw[dline, red] (d4) to ($(v3) !0.5! (v4)$) to [in=30, out=-90] (Pd);
				
				\draw[dline, green] (d1) to [in=30, out=-160] ($(v0) !0.5! (v1)$) to [in=50, out=-150] (dl);
				
				\draw[dline, blue] (d5) to [in=150, out=-20] ($(v4) !0.5! (v5)$) to [in=130, out=-30] (dr);
				
				\draw[dline, red] (d1) to (1,0.6);
				\draw[dline, green] (d2) to (1,0.6);
				
				\draw[dline, red] (d5) to (4,0.6);
				\draw[dline, blue] (d4) to (4,0.6);
				
				%					\draw[dline, red] (dl) to ($(v1) !0.5! (v2) !-0.5! (dl)$);
				
				%					\draw[dline, red] (dr) to ($(v3) !0.5! (v4) !-0.5! (dr)$);
				
				\draw[line, red] (v4) to ($(v4) !0.5! (v5)$);
				\draw[line, green] (v5) to ($(v4) !0.5! (v5)$);
				
				\draw[line, red] (v1) to ($(v1) !0.5! (v0)$);
				\draw[line, blue] (v0) to ($(v1) !0.5! (v0)$);
				
				\node[rnode, label=left:{$w_1$}] (w1) at (0,-0.7) {};
				\draw[line, blue]  (v1) to (w1);
				
%				\draw[dline, red] (dl) to ($(0,-0.7) !0.5! (v1)$);
%				\draw[dline, green] (Pd) to [in=-20.32, out=170] ($(0,-0.7) !0.5! (v1)$);
				
				\node[rnode, label=right:{$w_k$}] (wk) at (5,-0.7) {};
				\draw[line, green] (v4) to (wk);
				
%				\draw[dline, red] (dr) to ($(5,-0.7) !0.5! (v4)$);
%				\draw[dline, blue] (Pd) to [in=-160.68, out=10] ($(5,-0.7) !0.5! (v4)$);

				\draw[line, red] (1,-1.4) to (v1);
				\draw[line, red] (2.25,-1.4) to (v3);
				\draw[line, red] (3.75,-1.4) to (v3);
				
				\node[snode] (p1) at (0.6,-0.7) {};
				\node[snode] (p2) at (1.75,-0.7) {};
				\node[snode] (p3) at (3,-0.7) {};
				\node[snode] (p4) at (4,-0.7) {};
				
				\draw[dline, blue] (p1) to ($(v1) !0.5! (1,-1.4)$);
				\draw[dline, green] (p2) to ($(v1) !0.5! (1,-1.4)$);
				
				\draw[dline, blue] (p2) to ($(v3) !0.5! (2.25,-1.4)$);
				\draw[dline, green] (p3) to ($(v3) !0.5! (2.25,-1.4)$);
				
				\draw[dline, blue] (p3) to ($(v3) !0.5! (3.75,-1.4)$);
				\draw[dline, green] (p4) to ($(v3) !0.5! (3.75,-1.4)$);
				
				\draw[dline, green] (p1) to [in=-60, out=170] ($(w1) !0.3! (v1)$);
				\draw[dline, red] (dl) to [in=120, out=-60] ($(w1) !0.3! (v1)$);
				
				\draw[dline, blue] (p4) to [in=-130, out=0] ($(wk) !0.3! (v4)$);
				\draw[dline, red] (dr) to [in=50, out=-120] ($(wk) !0.3! (v4)$);
				
				\draw[dline, red] (d2) to ($(v1) !0.5! (v2)$) to [in=140, out=-90] (p2);
				\draw[dline, red] (d3) to ($(v2) !0.5! (v3)$) to [in=30, out=-90] (p2);
				\draw[dline, red] (d4) to ($(v3) !0.5! (v4)$) to [in=140, out=-90] (p4);
				
				\begin{pgfonlayer}{background}
					\draw[nline, line cap=round] (p1.center) to (p4.center);
				\end{pgfonlayer}
				
			\end{tikzpicture}
			\caption{Here, $w_1v_1$ is unidirected \blue-colored and $w_kv_k$ is unidirected \green-colored. }
		\end{subfigure}
		\begin{subfigure}{.8\textwidth}
			\centering
			\begin{tikzpicture}[scale=\textwidth/7cm]		
				\node[] (labelnode2) at (2.2,-0.2) {$P$};
				\node[] (labelnode3) at (1.5,-0.9) {$P'$};
				
				\node[rnode, label=left:{$v_1$}] (v1) at (1,0) {};
				\node[rnode] (v2) at (2,0) {};
				\node[rnode] (v3) at (3,0) {};
				\node[rnode, label=below:{$v_k$}] (v4) at (4,0) {};
				
				\node[rnode, label=right:{$v_{k+1}$}] (v5) at (5,1) {};
				\node[rnode, label=left:{$v_0$}] (v0) at (0,1) {};
				
				\node[snode, label=right:{$y'$}] (dr) at (5,0) {};
				\node[snode, label=left:{$y$}] (dl) at (0,0) {};
				
				\node[snode] (d1) at (0.7,0.6) {};
				\node[snode, label=above:{$x$}] (d2) at (1.5,0.6) {};
				\node[snode] (d3) at (2.5,0.6) {};
				\node[snode] (d4) at (3.5,0.6) {};
				\node[snode] (d5) at (4.3,0.6) {};
				
%				\node[snode, label=below:{$P'$}] (Pd) at (2.5,-0.9) {};

				\foreach \x/\y in {v1/v2, v2/v3, v3/v4}{
					\draw[line, green] (\x) to ($(\x) !0.5! (\y)$);
					\draw[line, blue] (\y) to ($(\x) !0.5! (\y)$);
				}
				
				\foreach \x/\y in {d2/d3, d3/d4}{
					\draw[dline, blue] (\x) to ($(\x) !0.5! (\y)$);
					\draw[dline, green] (\y) to ($(\x) !0.5! (\y)$);
				}
				
				\draw[line, red] (v2) to ($(v2) + (0,1.5)$);
				\draw[line, red] (v3) to ($(v3) + (0,1.5)$);
				
				\draw[line, blue] ($(v1) + (0,1.5)$) to (v1);
				\draw[line, green] ($(v4) + (0,1.5)$) to (v4);
				
%				\draw[dline, red] (d2) to ($(v1) !0.5! (v2)$) to [in=150, out=-90] (Pd);
%				\draw[dline, red] (d3) to (Pd);
%				\draw[dline, red] (d4) to ($(v3) !0.5! (v4)$) to [in=30, out=-90] (Pd);
				
				\draw[dline, green] (d1) to [in=30, out=-160] ($(v0) !0.5! (v1)$) to [in=50, out=-150] (dl);
				
				\draw[dline, blue] (d5) to [in=150, out=-20] ($(v4) !0.5! (v5)$) to [in=130, out=-30] (dr);
				
				\draw[dline, red] (d1) to (1,0.6);
				\draw[dline, green] (d2) to (1,0.6);
				
				\draw[dline, red] (d5) to (4,0.6);
				\draw[dline, blue] (d4) to (4,0.6);
				
				\draw[line, red] (v4) to ($(v4) !0.5! (v5)$);
				\draw[line, green] (v5) to ($(v4) !0.5! (v5)$);
				
				\draw[line, red] (v1) to ($(v1) !0.5! (v0)$);
				\draw[line, blue] (v0) to ($(v1) !0.5! (v0)$);
				
				\node[rnode, label=left:{$w_1$}] (w1) at (0,-0.7) {};
				\draw[line, blue]  (v1) to ($(w1) !0.5! (v1)$);
				\draw[line, red]  (w1) to ($(w1) !0.5! (v1)$);
				
%				\draw[dline, green] (Pd) to [in=-20.32, out=170] ($(0,-0.7) !0.5! (v1)$) to (dl);
				
				\node[rnode, label=right:{$w_k$}] (wk) at (5,-0.7) {};
				\draw[line, green]  (v4) to ($(wk) !0.5! (v4)$);
				\draw[line, red]  (wk) to ($(wk) !0.5! (v4)$);
				
%				\draw[dline, blue] (Pd) to [in=-160.68, out=10] ($(5,-0.7) !0.5! (v4)$) to (dr);
				
				\draw[line, red] (1,-1.4) to (v1);
				\draw[line, red] (2.25,-1.4) to (v3);
				\draw[line, red] (3.75,-1.4) to (v3);
				
				\node[snode] (p1) at (0.6,-0.7) {};
				\node[snode] (p2) at (1.75,-0.7) {};
				\node[snode] (p3) at (3,-0.7) {};
				\node[snode] (p4) at (4,-0.7) {};
				
				\draw[dline, blue] (p1) to ($(v1) !0.5! (1,-1.4)$);
				\draw[dline, green] (p2) to ($(v1) !0.5! (1,-1.4)$);
				
				\draw[dline, blue] (p2) to ($(v3) !0.5! (2.25,-1.4)$);
				\draw[dline, green] (p3) to ($(v3) !0.5! (2.25,-1.4)$);
				
				\draw[dline, blue] (p3) to ($(v3) !0.5! (3.75,-1.4)$);
				\draw[dline, green] (p4) to ($(v3) !0.5! (3.75,-1.4)$);
				
				\draw[dline, green] (p1) to [in=-60, out=90] ($(w1) !0.5! (v1)$) to [in=-30, out=120] (dl);
				
				\draw[dline, blue] (p4) to [in=-130, out=20] ($(wk) !0.5! (v4)$) to [in=-150, out=50] (dr);
				
				\draw[dline, red] (d2) to ($(v1) !0.5! (v2)$) to [in=140, out=-90] (p2);
				\draw[dline, red] (d3) to ($(v2) !0.5! (v3)$) to [in=30, out=-90] (p2);
				\draw[dline, red] (d4) to ($(v3) !0.5! (v4)$) to [in=140, out=-90] (p4);
				
				\begin{pgfonlayer}{background}
					\draw[nline, line cap=round] (p1.center) to (p4.center);
				\end{pgfonlayer}
				
			\end{tikzpicture}
			\caption{Here, $w_1v_1$ is \red-\blue-colored and $w_kv_k$ is \red-\green-colored.}
		\end{subfigure}
		\caption{Case~3 such that $v_1w_1$ and $v_kw_k$ are outgoing \blue-colored and \green-colored, respectively. In general, more than one incoming \red-colored edge at $v_1$ and $v_k$ may occur. $P'$ is highlighted in yellow.}
		\label{fig_third.2_case_breaking_cycle}
	\end{figure}
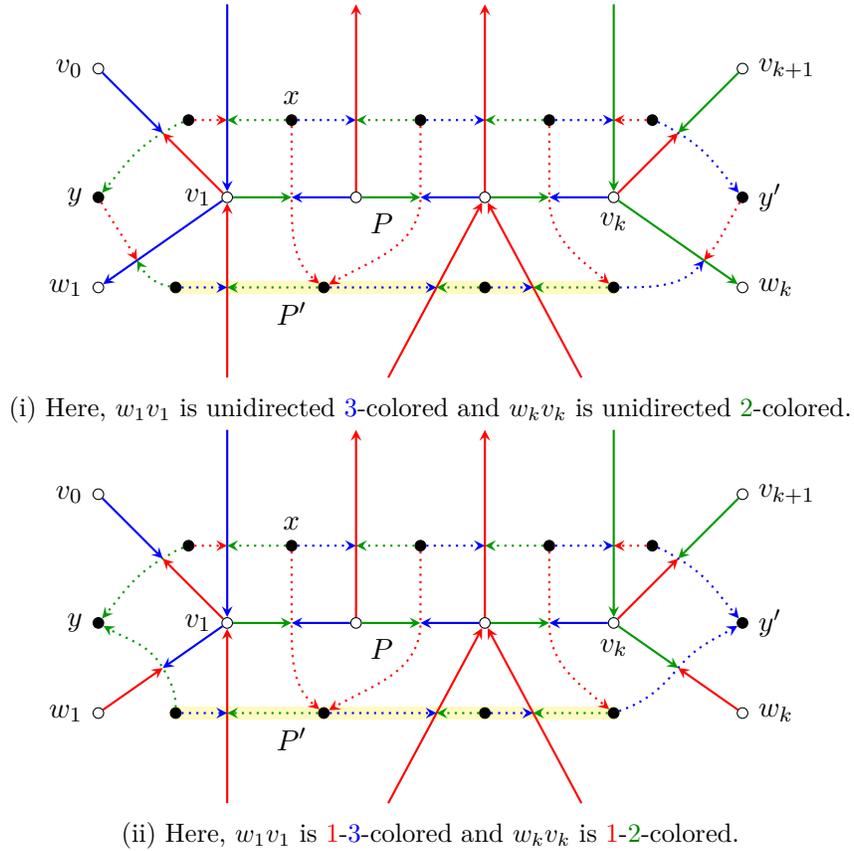
	
		In the remaining case, assume that $v_1$ does not have ingoing \green-colored edges and $v_k$ does not have ingoing \blue-colored edges. We prove that there is an edge of $P$, whose \red-colored dual edge is a parent edge and thus in $E(H(G^*))$ by~(H4). 
		
		Let $v_1w_1$ be the outgoing \blue-colored edge at $v_1$ and $v_kw_k$ be the outgoing \green-colored edge at $v_k$. Let $f$ be the face incident to $v_1$ and counterclockwise of $v_1v_2$ and let $x := f^*$ be its dual vertex. Let $y$ be the dual vertex such that $(v_0v_1)^*$ is incoming \green-colored at $y$ and let $y'$ be the dual vertex such that $(v_kv_{k+1})^*$ is incoming \blue-colored at $y'$ (see Figure~\ref{fig_third.2_case_breaking_cycle}).
		
		By Property~\ref{def:Schnyderwood}\ref{def:Schnyderwood3}, all edges in the clockwise sector from $v_iv_{i+1}$ to $v_iv_{i-1}$ around $v_i$ need to be unidirected incoming \red-colored, for all $i \in \{1,\ldots,k\}$. By Corollary~\ref{cor:crossingvertex}, then the dual compatible ordered path partition given by the maximal \green-\blue-colored paths contains a path $P'$ such that $(v_1v_2)^*, \ldots, (v_{k-1}v_k)^*$ are incoming \red-colored at vertices at $P'$. Also $y$ and $y'$ are the endpoints of the \green-colored and \blue-colored outgoing edge of $P'$, respectively (see Figure~\ref{fig_third.2_case_breaking_cycle}).
		
		Consider the path in $G^*$ from $y$ to $x$ that is incident to the face $v_1^*$ in clockwise direction. This path starts with an incoming \green-colored edge at $y$, and may contain further edges that are \red-\green-colored such that color~\red\ points towards $x$, as the corresponding primal edges incident to $v_1$ are unidirected and of color~\blue\ (see Figure~\ref{fig_third.2_case_breaking_cycle}). Hence, by Lemma~\ref{lem_compatible}\ref{lem_compatible_consistency}, the index of the dual ordered path partition (with maximal \green-\blue-colored paths) along this path increases; in particular, the index of $x$ is larger than the index of $y$. Thus, in the dual ordered path partition, the edge from $y$ to $P'$ is not a parent edge, as $x$ has a larger index than $y$ (which is still lower than the index of $P'$). 
		A symmetric argument implies that the edge from $y'$ to $P'$ is not a parent edge.
		
		Since, by Definition~\ref{def:Schnyderwood}\ref{def:Schnyderwood3} and Lemma~\ref{lem_compatible}\ref{lem_compatible_consistency}, all remaining incident edges of $P'$ that are not dual edges of $P$ lead to larger indices, only the dual edge of one of the edges of $P$ satisfies the requirements of a parent edge. So the parent edge of $P'$ is a dual of an edge of $P$ and thus in $E(H(G^*))$.
	\end{description}
	This completes the proof, as we showed that there is an edge $e$ in $C$ such that $e^* \in E(H(G^*))$.
\end{proof}

The main theorem follows by combining Lemmas~\ref{lem_deg_leq_5}, \ref{lem_dual_edge} and \ref{lem:cycle} as follows.

\begin{theorem}[Biedl \cite{B14}]
	Every 3-connected plane graph contains a spanning tree $T$ such that $T$ and its co-tree both have maximum degree 5.
\end{theorem}
\begin{proof}
	Our graph $H(G)$ has essentially the same properties (Lemmas~\ref{lem_deg_leq_5}, \ref{lem_dual_edge}, \ref{lem:cycle}) as the graph Biedl uses. So we can use straight away the proof of Biedl~\cite[Theorem~3]{B14} replacing her graph $H(G)$ with our graph $H(G)$. For the convenience of the reader, we reiterate her proof in our notation.
	
	We first argue that $H(G)$ is connected. Assume, for the sake of contradiction, that $H(G)$ is disconnected. Then there exists an edge-cut $Z$ with all cut-edges in $E(G) \setminus E(H(G))$. By Lemma~\ref{lem_dual_edge}, $Z^* \subseteq E(H^\circ(G^*))$. An edge-cut in a planar graph corresponds to the union of cycles in the dual graph \cite[Prop. 4.6.1]{Diestel2012}. So $Z^*$ contains a cycle $C$ of edges of $H^\circ(G^*)$. By Lemma~\ref{lem:cycle}, the dual of one edge of $C$ is in $H(G)$, which contradicts the definition of the cut. We conclude that $H(G)$ is connected.
	
	Let $H_0$ be the set of edges of $H(G)$ whose dual edges are not contained in $H^\circ(G^*)$. By Lemma~\ref{lem:cycle}, $H_0$ is a forest. Now we assign weights to the edges in order to compute a minimum weight spanning tree. The edges of $H_0$, $H(G) - H_0$ and the remaining edges of $G - H(G)$ get weight 0, 1 and $\infty$, respectively. Let $T$ be a minimum weight spanning tree of this instance. Since $H_0$ is a forest, all its edges are in $T$. As $H(G)$ is connected, no edge of $G - H(G)$ is in $T$. Hence, $T \subseteq H(G)$ and $T$ has maximum degree at most 5 by Lemma~\ref{lem_deg_leq_5}. The co-tree $\neg T^*$ consists of duals of edges of $G - H_0$. By definition of $H_0$ then these are all in $H^\circ(G^*)$. Thus, $\neg T^*$ has maximum degree 5 as well by Lemma~\ref{lem_deg_leq_5}.
\end{proof}

\section{Conclusion}
We gave an alternative proof of the result of Biedl~\cite{B14}. We are optimistic that the framework of Schnyder woods that we used for this proof turns out to be helpful in the quest towards a solution of Grünbaum's conjecture.

\bibliographystyle{abbrv}
\bibliography{paper}

\begin{thebibliography}{10}

\bibitem{Aerts2015}
N.~Aerts and S.~Felsner.
\newblock Straight-line triangle representations via {S}chnyder labelings.
\newblock {\em J. Graph Algorithms Appl.}, 19(1):467--505, 2015.

\bibitem{Alam2015}
M.~J. Alam, W.~Evans, S.~G. Kobourov, S.~Pupyrev, J.~Toeniskoetter, and
  T.~Ueckerdt.
\newblock Contact representations of graphs in {3D}.
\newblock In {\em Proceedings of the 14th International Symposium on Algorithms
  and Data Structures (WADS '15)}, volume 9214 of {\em Lecture Notes in
  Computer Science}, pages 14--27, 2015.
\newblock Technical Report accessible on arXiv:
  \href{https://arxiv.org/abs/1501.00304}{arxiv.org/abs/1501.00304}.

\bibitem{BBC11}
M.~Badent, U.~Brandes, and S.~Cornelsen.
\newblock More canonical ordering.
\newblock {\em J. Graph Algorithms Appl.}, 15(1):97--126, 2011.

\bibitem{Barnette1966}
D.~Barnette.
\newblock Trees in polyhedral graphs.
\newblock {\em Canadian Journal of Mathematics}, 18:731--736, 1966.

\bibitem{B14}
T.~Biedl.
\newblock Trees and co-trees with bounded degrees in planar 3-connected graphs.
\newblock In {\em Algorithm theory---{SWAT} 2014}, volume 8503 of {\em Lecture
  Notes in Comput. Sci.}, pages 62--73. Springer, Cham, 2014.

\bibitem{BTV99}
G.~{Di B}attista, R.~Tamassia, and L.~Vismara.
\newblock Output-sensitive reporting of disjoint paths.
\newblock {\em Algorithmica}, 23(4):302--340, 1999.

\bibitem{Diestel2012}
R.~Diestel.
\newblock {\em Graph theory}.
\newblock Graduate texts in mathematics 173. Springer, Berlin, 4th edition
  edition, 2012.

\bibitem{Felsner2001}
S.~Felsner.
\newblock Convex drawings of planar graphs and the order dimension of
  3-polytopes.
\newblock {\em Order}, 18(1):19--37, 2001.

\bibitem{F04GGAA}
S.~Felsner.
\newblock {\em Geometric graphs and arrangements}.
\newblock Advanced Lectures in Mathematics. Friedr. Vieweg \& Sohn, Wiesbaden,
  2004.
\newblock Some chapters from combinatorial geometry.

\bibitem{F04LSFPG}
S.~Felsner.
\newblock Lattice structures from planar graphs.
\newblock {\em Electron. J. Combin.}, 11(1):Research Paper 15, 24, 2004.

\bibitem{Gruenbaum1970}
B.~Gr\"unbaum.
\newblock Polytopes, graphs, and complexes.
\newblock {\em Bulletin of the American Mathematical Society},
  76(6):1131--1201, 1970.

\bibitem{Kant1992}
G.~Kant.
\newblock Drawing planar graphs using the lmc-ordering.
\newblock In {\em Proceedings of the 33rd Annual Symposium on Foundations of
  Computer Science (FOCS'92)}, pages 101--110, 1992.

\bibitem{Kant1996}
G.~Kant.
\newblock Drawing planar graphs using the canonical ordering.
\newblock {\em Algorithmica}, 16(1):4--32, 1996.

\bibitem{Schnyder1990}
W.~Schnyder.
\newblock Embedding planar graphs on the grid.
\newblock In {\em Proceedings of the first annual ACM-SIAM symposium on
  Discrete algorithms}, pages 138--148. Society for Industrial and Applied
  Mathematics, 1990.

\end{thebibliography}
\end{document}